\documentclass[11pt]{amsart}

\usepackage[centertags]{amsmath}
\usepackage{amsthm,amsfonts}

\usepackage{setspace}
\usepackage{amssymb}
\usepackage{epsfig}
\usepackage{amssymb,latexsym}
\usepackage{mathdots}
\usepackage[usenames,dvipsnames]{color}

\begin{document}
%%%%%%%%%%%%%%%%%%%% Text italic %%%%%%%%%%%%%%%%%%%%%%%%%%%%
\theoremstyle{plain}
\newtheorem{theorem}{Theorem}[section]
\newtheorem{lemma}[theorem]{Lemma}
\newtheorem{corollary}[theorem]{Corollary}
\newtheorem{proposition}[theorem]{Proposition}
\newtheorem{example}[theorem]{Example}
\newtheorem{examples}[theorem]{Examples}
\newtheorem{question}[theorem]{Question}
\newtheorem{conjecture}[theorem]{Conjecture}
%%%%%%%%%%%%%%%%%%%% Text roman %%%%%%%%%%%%%%%%%%%%%%%%%%%%%
\theoremstyle{definition}
\newtheorem{notations}[theorem]{Notations}
\newtheorem{notation}[theorem]{Notation}
\newtheorem{remark}[theorem]{Remark}
\newtheorem{remarks}[theorem]{Remarks}
\newtheorem{definition}[theorem]{Definition}
\newtheorem{claim}[theorem]{Claim}
\newtheorem{problem}[theorem]{Problem}
\newtheorem{assumption}[theorem]{Assumption}
\numberwithin{equation}{section}
\newtheorem{examplerm}[theorem]{Example}
\newtheorem{propositionrm}[theorem]{Proposition}
%%%%%%%%%%%%%%%%%%%% einige Abkuerzungen  %%%%%%%%%%%%%%%%%%%

\newcommand{\binomial}[2]{\left(\begin{array}{c}#1\\#2\end{array}\right)}
\newcommand{\zar}{{\rm zar}}
\newcommand{\an}{{\rm an}}
\newcommand{\red}{{\rm red}}
\newcommand{\codim}{{\rm codim}}
\newcommand{\rank}{{\rm rank}}
\newcommand{\Pic}{{\rm Pic}}
\newcommand{\Div}{{\rm Div}}
\newcommand{\Hom}{{\rm Hom}}
\newcommand{\im}{{\rm im}}
\newcommand{\Spec}{{\rm Spec}}
\newcommand{\sing}{{\rm sing}}
\newcommand{\reg}{{\rm reg}}
\newcommand{\Char}{{\rm char}}
\newcommand{\Tr}{{\rm Tr}}
\newcommand{\tr}{{\rm tr}}
\newcommand{\supp}{{\rm supp}}
\newcommand{\Gal}{{\rm Gal}}
\newcommand{\Min}{{\rm Min \ }}
\newcommand{\Max}{{\rm Max \ }}
\newcommand{\Span}{{\rm Span  }}

\newcommand{\Frob}{{\rm Frob}}
\newcommand{\lcm}{{\rm lcm}}

\newcommand{\soplus}[1]{\stackrel{#1}{\oplus}}
\newcommand{\dlog}{{\rm dlog}\,}    % For dlog
\newcommand{\limdir}[1]{{\displaystyle{\mathop{\rm
lim}_{\buildrel\longrightarrow\over{#1}}}}\,}
\newcommand{\liminv}[1]{{\displaystyle{\mathop{\rm
lim}_{\buildrel\longleftarrow\over{#1}}}}\,}
\newcommand{\boxtensor}{{\Box\kern-9.03pt\raise1.42pt\hbox{$\times$}}}
\newcommand{\sext}{\mbox{${\mathcal E}xt\,$}}
\newcommand{\shom}{\mbox{${\mathcal H}om\,$}}
\newcommand{\coker}{{\rm coker}\,}
\renewcommand{\iff}{\mbox{ $\Longleftrightarrow$ }}
\newcommand{\onto}{\mbox{$\,\>>>\hspace{-.5cm}\to\hspace{.15cm}$}}

\newenvironment{pf}{\noindent\textbf{Proof.}\quad}{\hfill{$\Box$}}

\newcommand{\sA}{{\mathcal A}}
\newcommand{\sB}{{\mathcal B}}
\newcommand{\sC}{{\mathcal C}}
\newcommand{\sD}{{\mathcal D}}
\newcommand{\sE}{{\mathcal E}}
\newcommand{\sF}{{\mathcal F}}
\newcommand{\sG}{{\mathcal G}}
\newcommand{\sH}{{\mathcal H}}
\newcommand{\sI}{{\mathcal I}}
\newcommand{\sJ}{{\mathcal J}}
\newcommand{\sK}{{\mathcal K}}
\newcommand{\sL}{{\mathcal L}}
\newcommand{\sM}{{\mathcal M}}
\newcommand{\sN}{{\mathcal N}}
\newcommand{\sO}{{\mathcal O}}
\newcommand{\sP}{{\mathcal P}}
\newcommand{\sQ}{{\mathcal Q}}
\newcommand{\sR}{{\mathcal R}}
\newcommand{\sS}{{\mathcal S}}
\newcommand{\sT}{{\mathcal T}}
\newcommand{\sU}{{\mathcal U}}
\newcommand{\sV}{{\mathcal V}}
\newcommand{\sW}{{\mathcal W}}
\newcommand{\sX}{{\mathcal X}}
\newcommand{\sY}{{\mathcal Y}}
\newcommand{\sZ}{{\mathcal Z}}

% Sonderbuchstaben mit Doppellinie
\newcommand{\A}{{\mathbb A}}
\newcommand{\B}{{\mathbb B}}
\newcommand{\C}{{\mathbb C}}
\newcommand{\D}{{\mathbb D}}
\newcommand{\E}{{\mathbb E}}
\newcommand{\F}{{\mathbb F}}
\newcommand{\G}{{\mathbb G}}
\newcommand{\HH}{{\mathbb H}}
\newcommand{\I}{{\mathbb I}}
\newcommand{\J}{{\mathbb J}}
\newcommand{\M}{{\mathbb M}}
\newcommand{\N}{{\mathbb N}}
\renewcommand{\P}{{\mathbb P}}
\newcommand{\Q}{{\mathbb Q}}
\newcommand{\T}{{\mathbb T}}
\newcommand{\U}{{\mathbb U}}
\newcommand{\V}{{\mathbb V}}
\newcommand{\W}{{\mathbb W}}
\newcommand{\X}{{\mathbb X}}
\newcommand{\Y}{{\mathbb Y}}
\newcommand{\Z}{{\mathbb Z}}

%%%%%%%%%%%%%%%%%%%%%%%%%%%%%%%%%%%%%%%
\newcommand{\be}{\begin{eqnarray}}
\newcommand{\ee}{\end{eqnarray}}
\newcommand{\nn}{{\nonumber}}
\newcommand{\dd}{\displaystyle}
\newcommand{\ra}{\rightarrow}
\newcommand{\bigmid}[1][12]{\mathrel{\left| \rule{0pt}{#1pt}\right.}}
\newcommand{\cl}{${\rm \ell}$}
\newcommand{\clp}{${\rm \ell^\prime}$}

%%%%%%%%%%%%%%%%%%%%%%%%%%%%%%%%%%%%%%%%
%\newcommand{\id}[1]{\left< #1\right>}
\newcommand{\id}[1]{\langle #1\rangle}
\newcommand{\R}{GR(p^a,m)}
\newcommand{\Rf}{\frac{\R[x]}{\id{x^{2p^s}-1}}}
\newcommand{\fR}{\R}
\newcommand{\nck}[2]{{#1 \choose #2}}
\newcommand{\mySkip}{\bigskip}

\newcommand{\fm}[1][n]{ \ensuremath{f(x_1,\dots, x_{#1})}}
\newcommand{\gm}[1][n]{ \ensuremath{g(x_1,\dots, x_{#1})}}
\newcommand{\FqR}[1][n]{ {\mathbb F}_q[x_1, \dots , x_{#1}] }
%\newcommand{\IdQ}[1][r]{\langle x_1, \dots , x_{#1} \rangle }
%%%%%%%%%%%%%%%%%%%%%%%%%%%%%%%%%%%%%%%%%%%%%%%%%%%%%%%%%%%%%%
%%%%%%%%%%%%%%%%%%%%%%%%%%%%%%%%%%%%%%%%%%%%%%%%%%%%%%%%%%%%%%

%%%% C o m m e n t    m a c r o s %%%%%%%%%%%%%%%%%%%%%%%%%%%%%

%\newcommand{\Edgar}{ \color{red}{Edgar's Comment: #1} }
\newcommand{\Edgar}[1]{{\color{red}{Edgar's Comment: #1}} }
\newcommand{\Steve}[1]{{\color{blue}{Steve's Comment: #1}} }
\newcommand{\Ferruh}[1]{{\color{ForestGreen}{Ferruh's Comment: #1}} }
%%%%%%%%%%%%%%%%%%%%%%%%%%%%%%%%%%%%%%%%%%%%%%%%%%%%%%%%%%%%%%%

\title[Monomial-like codes]{Monomial-like codes}

\author[Edgar Mart\'{i}nez-Moro, Hakan \"{O}zadam, Ferruh \"{O}zbudak, Steve Szabo]{Edgar Mart\'{i}nez-Moro$^1$, Hakan \"{O}zadam$^2$, Ferruh \"{O}zbudak$^2$, Steve Szabo$^3$}
\maketitle

\begin{center}
$^1$
Departamento de Matem\'{a}tica Aplicada\\
Universidad de Valladolid, Campus Duques de Soria, 42003 Soria, Castilla, Spain\\
e-mail: edgar@maf.uva.es\\ \ \\
$^2$\ Department of Mathematics and Institute of Applied Mathematics \\
Middle East Technical University, \.{I}n\"on\"u Bulvar{\i}, 06531, Ankara, Turkey \\
e-mail: \{ozhakan,ozbudak\}@metu.edu.tr\\ \ \\
$^3$\ Department of Mathematics\\
Ohio University, Athens, Ohio, 45701, USA\\
szabo@math.ohiou.edu
\end{center}

%\vspace{0.5cm}

\abstract
   As a generalization of cyclic codes of length $p^s$\ over $\F_{p^a}$,
   we study $n$-dimensional cyclic codes of length $p^{s_1} \times \cdots \times p^{s_n}$\
   over $\F_{p^a}$\ generated by a single  ``monomial''. 
   Namely, we study multi-variable cyclic codes of the form
   $\id{(x_1 - 1)^{i_1} \cdots (x_n - 1)^{i_n}} \subset \frac{\FqR}{\id{ x_1^{p^{s_1}}-1, \dots , x_n^{p^{s_n}}-1  }}$.
   We call such codes \textit{monomial-like codes}.
   We show that these codes arise from the product of certain single variable codes 
   and we determine their minimum Hamming distance.
   We determine the dual of monomial-like codes yielding a parity check matrix.
   We  also present an alternative way of constructing a parity check matrix using the Hasse derivative.
   We study  the weight hierarchy of certain monomial like codes.
   We simplify an expression that gives us the  weight hierarchy of these codes.
\endabstract

\vspace{0.5cm}


%\vspace{0.5cm}
%\vspace{0.5cm}
%%%%%%%%%%%%%%%%%%%%%%%%%%%%%%%%%%%%%%%%%%%%%%%%%%%%%%%%%%%%%%%%%%%%%
%%%%%%%%%%%%%%%&&&    I N T R O D U C T I O N %%%%%%%%%%%%%%%%%%%%%%%
%%%%%%%%%%%%%%%%%%%%%%%%%%%%%%%%%%%%%%%%%%%%%%%%%%%%%%%%%%%%%%%%%%%%%
\section{Introduction}
\label{Section.Introduction}

Cyclic codes are said  to be repeated-root when the codeword length and the 
characteristic of the alphabet are not coprime. 
In some cases repeated-root cyclic codes have the following interesting properties.
Massey et. al. have shown in \cite{MASSEY_1973}
that cyclic codes of length $p$\ over a finite field of characteristic $p$\ are optimal.
There also exist infinite families of repeated-root cyclic codes in even characteristic 
according to the results of \cite{vanlint_1991}.
It was pointed out in \cite{MASSEY_1973} that some repeated-root cyclic codes can be decoded
using a very simple circuitry.
Among the studies on repeated-root cyclic codes are \cite{castagnoli_1991}, \cite{dinh_2008}, \cite{lopez_2009}, 
\cite{Martinez-Moro_2007}, 
 \cite{MASSEY_1973}, \cite{ozadam_2009_2} and \cite{vanlint_1991}.

Contrary to the simple-root case, there are repeated root cyclic codes of the form $\id{f^i(x)}$\
where $i>1$. Specifically, all cyclic codes of length $p^s$\ over a finite field of characteristic
$p$\ are generated by a single ``monomial'' of the form $(x-1)^i$, where $0 \le i \le p^s$\
(c.f. \cite{dinh_2008} and \cite{ozadam_2009_2}).
In this paper, as a generalization of these codes to many variables,
we study cyclic codes of the form
\be\nn
  \id{(x_1 - 1)^{i_1} \cdots (x_n -1)^{i_n}}\subset \frac{\F_{p^a}[x_1, \dots , x_n]}{\id{x_1^{p^{s_1}} -1, \dots , x_n^{p^{s_n}}-1 }}.
\ee
In other words, we study $n$-dimensional cyclic codes of length $p^{s_1} \times \cdots \times p^{s_n}$,
generated by a single ``monomial'', over a finite field of characteristic $p$.
We call these codes ``monomial-like'' codes.
After exploring some properties of the ambient space $\frac{\F_{p^a}[x_1, \dots , x_n]}{\id{x_1^{p^{s_1}}, \dots , x_n^{p^{s_n}}}}$,
we show that monomial-like codes arise from product codes.
More precisely, we show that multi-variable monomial-like codes are actually the product of
one-variable monomial-like codes.
This enables us to express the minimum Hamming distance of monomial-like codes
as a product of the minimum Hamming distance of cyclic codes of length $p^s$\
which was computed in  \cite{dinh_2008} and \cite{ozadam_2009_2}.
In addition to this, we determine the dual of monomial-like codes which also yields a parity check matrix 
for monomial like codes.

The weight hierarchy of  linear codes was introduced in \cite{Helleseth_Gen_Ham_Weight_First} and
\cite{Wei_Generaized_HW_1991}.
For some application motives, the weight hierarchy is considered as an important property of a liner code.
We simplify an expression, which was conjectured in \cite{WEI_YANG_Conjecture} and proved in
\cite{schaathun_2000}, for certain monomial like codes.
We obtain a simplified expression that gives the weight hierarchy of the monomial like codes
which are products of cyclic codes of length $p$\ over $\F_{p^a}$.   

In \cite{castagnoli_1991}, the authors show how to construct a parity check matrix for
repeated-root cyclic codes in one variable.
This construction is based on the Hasse derivative and the repeated-root factor test.
When the codeword length is a power of $p$, their construction applies to
monomial-like codes in one variable.
We generalize the repeated-root factor test and the construction of the parity check matrix
to monomial-like codes in many variables.

This paper is organized as follows.
First we introduce some notation, give some definitions and prove some structural
properties of the ambient space of monomial-like codes in Section \ref{Section.Preliminaries}.
In Section \ref{Section.Monomial.Like.Codes}, we define monomial-like codes. 
We show that these codes arise from product codes and we determine their Hamming distance.
We describe the dual of monomial-like codes which yields a parity check matrix for these codes.
In Section \ref{Section.Generalized.Ham.W.Prod.Codes}, we study the generalized Hamming weight of some product codes
and simplify an expression which gives the weight hierarchy of certain monomial-like codes.
In Section \ref{Section.Hasse.Der}, we explain how to construct a parity check matrix for
monomial-like codes using the Hasse derivative.

%%%%%%%%%%%%%%%%%%%%%%%%%%%%%%%%%%%%%%%%%%%%%%%%%%%%%%%%%%%%%%%%%%%%%
%%%%%%%%%%%%%%%&   P R E L I M I N A R I E S %%%%%%%%%%%%%%%%%%%%%%%
%%%%%%%%%%%%%%%%%%%%%%%%%%%%%%%%%%%%%%%%%%%%%%%%%%%%%%%%%%%%%%%%%%%%%
\section{The Ambient Space}
\label{Section.Preliminaries}
Throughout the paper, we consider the  finite ring
\be
  \label{Definition.Ambient.Space}
  \sR = \frac{\FqR}{\id{x_1^{p^{s_1}} - 1, x_2^{p^{s_2}} - 1,\dots, x_n^{p^{s_n}} - 1} }
\ee
as the ambient space of the codes to be studied unless stated otherwise.
We define
\be
    L = \{ (i_1,i_2, \dots, i_n): \quad 0 \le i_j < p^{s_j},\quad i_j\in \Z \quad \mbox{for all}\quad 1 \le j \le n \}.\nn
\ee

The elements of $\sR$\ can be identified uniquely with the polynomials of the form
\be\label{Preliminaries.Polynomial.f.as.Codeword}
   \fm = \sum_{(i_1,i_2,\dots, i_n) \in L}f_{(i_1,i_2,\dots, i_n)}x_1^{i_1}x_2^{i_2}\cdots x_n^{i_n},\nn
\ee
so throughout the paper, we identify the equivalence class 
$\fm + \id{x_1^{p^{s_1}} - 1, x_2^{p^{s_2}} - 1,\dots, x_n^{p^{s_n}} - 1}$
with the polynomial $\fm$.
The $n$-dimensional cyclic codes over $\F_q$\ of length $p^{s_1} \times p^{s_2} \times \cdots \times p^{s_n}$\
are exactly the ideals of $\sR$\ where we identify each codeword
$\displaystyle (f_{(i_1,i_2,\dots, i_n)})_{(i_1,i_2,\dots, i_n)\in L}$\ 
with the polynomial
$\fm$\ via a fixed monomial ordering.
The support of $\fm$\ is the set
\be
    \sG (f) = \{ (i_1,i_2,\dots, i_n) \in L : \quad f_{(i_1,i_2,\dots, i_n)} \neq 0 \}, \nn
\ee
and the Hamming weight of $\fm$\ is defined as
\be
    w_H( \fm ) = |\sG (f)|,\nn
\ee
i.e., the number of nonzero coefficients of $\fm$.
The minimum Hamming distance of a code $C$\ is defined as 
$$d_H(C) = \min \{w_H( \fm ): \quad \fm \in C\setminus \{ 0 \}  \}.$$

\begin{lemma}
\label{Preliminaries.Lemma.R.is.local}
  $\sR$\ is a local ring with the maximal ideal $M = \id{x_1-1, x_2-1, \dots, x_n - 1}$.
\end{lemma}
\begin{proof}
 Let $\fm \in \sR$.
 Using the substitution $x_{\ell} = (x_{\ell}-1) + 1$, for all $1 \le \ell \le r$, we can express $\fm$\ as
 \be
  \fm = \sum c_{i_1,i_2,\dots,i_n}x_1^{i_1}x_2^{i_2}\cdots x_n^{i_n} =
    \sum d_{i_1,i_2,\dots,i_n}(x_1-1)^{i_1}(x_2-1)^{i_2}\cdots (x_n-1)^{i_n}.\nn
 \ee
 If $d_{0,0,\dots,0} \neq 0$, then $\fm = f_0(x_1,x_2,\dots,x_n) + d_{0,0,\dots,0}$\
 for some $f_0(x_1,x_2,\dots,x_n) \in \id{x_1 - 1, x_2 - 1,\dots, x_n - 1}$.
 Since $x_{\ell}-1$\ are nilpotent, for all $1 \le \ell \le n$, $f_0(x_1,x_2,\dots,x_n)$\
 is also a nilpotent element and therefore,
 being a sum of a nilpotent element and a unit, $f(x_1,x_2,\dots,x_n)$\ is a unit.
 In other words, $\sR \setminus \{ \id{x_1 - 1, x_2 - 1,\dots, x_n - 1} \}$\
 consists of exactly the units of $\sR$.
 This implies that $\sR$\ is a local ring with the maximal ideal
 $\id{x_1 - 1, x_2 - 1,\dots, x_n - 1}$.
\end{proof}

\begin{remark}
  \label{Preliminaries.Remark.Structure.Complicated}
 Not all the ideals of $\sR$\ are of the form $\id{(x_1-1)^{i_1},\dots , (x_n - 1)^{i_n}}$. 
 As a counter-example, we consider
 $$ \hat{\sR} = \frac{\F_q[x,y]}{\id{x^{p^{s_1}} -1,y^{p^{s_2}} -1}}, $$
 $\hat{I} = \id{x^{p^{s_1}} -1,y^{p^{s_2}} -1} $ and let $J = \id{(x - 1)(y -1)} + \hat{I}$.
 Suppose that there  exist $p^{s_1} > m > 0$\ and $p^{s_2} > n > 0$\ such that
 \be
  J = \id{(x-1)^m,(y-1)^n} + \hat{I}.\nn
 \ee
 Then $(x-1)^m + \hat{I} \in J = \id{ (x-1)(y-1) } + \hat{I}$. So,
  for some $g(x,y)\in \F_q[x,y]$, we have, $(x-1)^m - g(x,y)(x-1)(y-1) \in
      \hat{I} = \id{x^{p^{s_1}} -1,y^{p^{s_2}} -1} $. Therefore
 \be
  \label{Preliminaries.Remark.Structure.Complicated.eq}
 (x-1)^m - g(x,y)(x-1)(y-1) = \alpha_1(x,y)(x-1)^{p^{s_1}} + \alpha_2(x,y)(y-1)^{p^{s_2}}
 \ee
 for some $\alpha_1(x,y), \alpha_2(x,y) \in \F_q[x,y]$. Evaluating both sides of
 (\ref{Preliminaries.Remark.Structure.Complicated.eq}) at $y=1$, we get
 \be
  (x-1)^m = \alpha_1(x,1)(x-1)^{p^{s_1}}.\nn
 \ee
 This is a contradiction because $m < p^{s_1}$.
\end{remark}

\begin{remark}
 \label{Remark.Isomorhpism.Basic.Change.of.Variable}
 We have the ring isomorphism
 \be
  \frac{\FqR }{ \id{x_1^{p^{s_1}} - 1, x_2^{p^{s_2}} - 1,\dots, x_n^{p^{s_n}} - 1 }}
    \cong \frac{\F_q[y_1,\dots , y_n]}{ \id{y_1^{p^{s_1}}, y_2^{p^{s_2}},\dots, y_n^{p^{s_n}} }} \nn
 \ee
 where the isomorphism is established by sending $x_i-1$\ to $y_i$.
\end{remark}

Let $p$\ be an odd prime and let $GR(p^a,m)$\ be the Galois ring of
characteristic $p^a$\ with $p^{am}$\ elements. Let $N$\ be an odd
integer. By \cite[Proposition 5.1]{dinh_2004}, we know that the
rings
$$ \frac{GR(p^a,m)[x]}{ \id{x^N - 1} }\quad \mbox{and} \quad \frac{GR(p^a,m)[x]}{ \id{x^N +1} } $$
are isomorphic. This is generalized to multi-variable constacyclic codes in the next lemma.
\begin{lemma}\label{Preliminaries.Isomorphism}
  Let $c_1, c_2, \dots ,c_n \in GR(p^a,m)^{*}$\ be some units of $GR(p^a,m)$.
  Let $k_1,k_2,\dots, k_n$\ be odd positive integers.
 The map
 $$ \xi: \frac{GR(p^a,m)[x_1,x_2,\dots, x_n]}{\id{ x_1^{k_1}-1, x_2^{k_2}-1,\dots , x_n^{k_n}-1 } } \rightarrow
      \frac{GR(p^a,m)[x_1,x_2,\dots, x_n]}{\id{ x_1^{n_1}-c_1^{k_1}, x_2^{k_2}-c_2^{n_2},\dots , x_r^{k_r}-c_n^{k_n} } }$$
  defined by
  $$ f(x_1,x_2,\dots, x_n) \mapsto f( c_1^{-1}x_1,c_2^{-1}x_2, \dots , c_n^{-1}x_n )$$
  is a ring isomorphism.
\end{lemma}
\begin{proof}
 Let $I = \id{ x_1^{k_1}-1,x_2^{k_2} -1,\dots ,x_n^{k_n}-1 }$\ and
 $J = \id{ x_1^{k_1}- c_1^{k_1},x^{k_2}- c_2^{k_2},\dots , x_n^{k_n}- c_n^{k_n} } $.
 For every $k_i$, we have $ (c_i^{-1}x_i)^{k_i} - 1 = c_i^{-k_i}(x_i^{k_i} - c_i^{k_i}) $.
 Therefore
 \be
    f(x_1,x_2,\dots , x_n) \equiv g(x_1,x_2,\dots , x_n) \mod I \nn
 \ee
 if and only if there are polynomials $h_e(x_1,x_2, \dots , x_n), e \in \{ 1,2,\dots ,n \}$\
 such that
 \be
    && f(x_1,x_2,\dots , x_n) - g(x_1,x_2,\dots , x_n) \nn \\
        & = & h_1(x_1,x_2,\dots , x_n)(x_1^{k_1} - 1)
        h_2(x_1,x_2,\dots , x_n)(x_2^{k_2} - 1) + \cdots + h_n(x_1,x_2,\dots , x_n)(x_n^{k_n} - 1)\nn
 \ee
 if and only if
 \be
    &&f(c_1^{-1}x_1, c_2^{-1}x_2,\dots , c_n^{-1} x_n) - g(c_1^{-1}x_1 , c_2^{-1}x_2,\dots , c_n^{-1} x_n)\nn\\
    & = & h_1(c_1^{-1}x_1, c_2^{-1} x_2,\dots , c_n^{-1} x_n)((c_1^{-1}x_1)^{k_1} - 1) +
        h_2(c_1^{-1} x_1, c_2^{-1} x_2,\dots , c_n^{-1} x_n)(( c_2^{-1} x_2)^{k_2} - 1)  \nn\\
        && \; + \cdots + h_r(c_1^{-1} x_1, c_2^{-1} x_2,\dots , c_n^{-1} x_n)(( c_n^{-1} x_n)^{k_n} - 1)\nn\\
    & = & h_1(c_1^{-1}x_1, c_2^{-1} x_2,\dots , c_n^{-1} x_n) c_1^{-k_1} (x_1^{k_1} - c_1^{k_1})
        + h_2(c_1^{-1}x_1, c_2^{-1} x_2,\dots , c_n^{-1} x_n) c_2^{-k_2} (x_2^{k_2} - c_2^{k_2})\nn\\
    && \; + \cdots + h_n(c_1^{-1}x_1, c_2^{-1} x_2,\dots , c_n^{-1} x_n) c_n^{-k_n} (x_n^{k_n} - c_n^{k_n})\nn
 \ee
 if and only if
 \be
    f(x_1,x_2,\dots , x_n) \equiv g(x_1,x_2,\dots , x_n) \mod J. \nn
 \ee
 This implies that $\xi$\ is well-defined and $\xi$\ is injective.
 The fact that $\xi$\ respects addition is obvious. It is also easy to see that
 \be \label{Preliminaries.Isomorphism.Xi.Respects.Monomial.Mult}
    \xi (ax_1^{i_1}x_2^{i_2}\cdots x_n^{i_n} f(x_1,x_2, \dots ,x_n) ) =
    \xi( ax_1^{i_1}x_2^{i_2}\cdots x_n^{i_n} ) \xi ( f(x_1,x_2, \dots ,x_n) ).
 \ee
 Together with the fact that $\xi$\ is linear, (\ref{Preliminaries.Isomorphism.Xi.Respects.Monomial.Mult})
 implies that $\xi (f(x_1,x_2, \dots , x_n) \cdot g(x_1,x_2, \dots , x_n)) =
 \xi(f(x_1,x_2, \dots , x_n)) \xi(g(x_1,x_2, \dots , x_n)) $, for every
 $$f(x_1,x_2, \dots , x_n), g(x_1,x_2, \dots , x_n) \in \dd 
  \frac{GR(p^a,m)[x_1,x_2,\dots, x_n]}{\id{ x_1^{k_1}-1, x_2^{k_2}-1,\dots , x_n^{k_n}-1 } }.$$
 Thus $\xi$\ is a ring homomorphism. For every
 $$ h(x_1,x_2, \dots, x_n) \in \dd 
  \frac{GR(p^a,m)[x_1,x_2,\dots, x_n]}{\id{ x_1^{k_1}-c_1^{k_1}, x_2^{k_2}-c_2^{k_2},\dots , x_n^{k_n}-c_n^{k_n} } }, $$
 we have $\xi ( h(c_1x_1, c_2x_2, \dots, c_nx_n) ) = h(x_1,x_2, \dots, x_n)$. Hence $\xi$\
 is onto. Thus $\xi$\ is an isomorphism.
\end{proof}

Lemma \ref{Preliminaries.Isomorphism} tells us that, in our case,
we can work with negacyclic codes instead of cyclic codes.
More precisely, we have the following corollary.

\begin{corollary}
  \label{Corollary.Isomorphism.of.Cyclic.and.Negacyclic.Codes}
 The rings
 $$\dd \frac{\F_q[x_1,x_2,\dots , x_n]}{\id{ x_1^{p^{s_1}}-1, x_2^{p^{s_2}}-1,\dots , x_n^{p^{s_n}}-1 }} \quad
 \mbox{and} \quad \dd \frac{\F_q[x_1,x_2,\dots , x_n]}{\id{ x_1^{p^{s_1}} + 1, x_2^{p^{s_2}} + 1,\dots , x_n^{p^{s_n}} + 1 }}$$
 are isomorphic, where the isomorphism is established by sending each $x_i$\ to $-x_i$.
 In even characteristic, these rings are exactly the same.
\end{corollary}

Thus, for $C_1 = \id{ (x_1 - 1)^{i_1},(x_2 - 1)^{i_2}, \dots , (x_n - 1)^{i_n} }
\subset \displaystyle \frac{\F_q[x_1,x_2,\dots , x_n]}
{ \id{ x_1^{p^{s_1}} - 1, x_2^{p^{s_2}} - 1, \dots , x_n^{p^{s_n}} - 1 } } $\ and
$C_2 =  \id{ (x_1 + 1)^{i_1},(x_2 + 1)^{i_2}, \dots , (x_n + 1)^{i_n} }
\subset \displaystyle \frac{\F_q[x_1,x_2,\dots , x_n]}
{ \id{ x_1^{p^{s_1}} + 1, x_2^{p^{s_2}} + 1, \dots , x_n^{p^{s_n}} + 1 } }$,
$C_1$\ and $C_2$\ have the same distance distribution and, consequently, have the same minimum Hamming distance.

%%%%%%%%%%%%%%%%%%%%%%%%%%%%%%%%%%%%%%%%%%%%%%%%%%%%%%%%%%%%%%%%%%%%%%%%%%%%%%%%%%%
%%%%%%%%%%%%%%% M O N O M I A L - L I K E    C O D E S  %%%%%%%%%%%%%%%%%%%%%%%%%%%
%%%%%%%%%%%%%%%%%%%%%%%%%%%%%%%%%%%%%%%%%%%%%%%%%%%%%%%%%%%%%%%%%%%%%%%%%%%%%%%%%%%
\section{Monomial-like codes}
\label{Section.Monomial.Like.Codes}

The elements of $\FqR$\ are $\F_q$-linear combinations of monomials
$x_1^{\alpha_1}x_2^{\alpha_2}\dots x_n^{\alpha_n}$.
From this perspective, one can say that monomials are building blocks of polynomials.
Analogously, as a consequence of Lemma \ref{Preliminaries.Lemma.R.is.local}, the elements of
$\sR$\
are $\F_q$-linear combinations of the terms $(x_1 - 1)^{\alpha_1}(x_2 - 1)^{\alpha_2} \dots (x_n - 1)^{\alpha_n}$,
where $(\alpha_1, \dots , \alpha_n) \neq (0, \dots , 0)$.
So, as was done in \cite{Lakatos_AAECC}, we call the terms 
$(x_1 - 1)^{\alpha_1}(x_2 - 1)^{\alpha_2} \dots (x_n - 1)^{\alpha_n}$\ as ``monomials'' and
ideals generated by monomials as  ``monomial ideals''. 
``Monomial ideals'' of $\frac{\FqR}{ \id{x_1^p -1, \dots, x_n^p -1} }$\ were studied in \cite{Lakatos_AAECC}.
We concentrate on a special class of monomial ideals
that are generated by a single monomial, in a more general ambient space. 
We call such ideals as ``\textit{monomial-like ideals}'' and
the corresponding codes as ``\textit{monomial-like codes}''.
Namely, monomial-like codes are of the form 
$ C = \id{ (x_1 - 1)^{\alpha_1}(x_2 - 1)^{\alpha_2} \cdots (x_n - 1)^{\alpha_n} } \subset \sR$.

Our aim is to determine the minimum Hamming distance of monomial-like codes.
Let $C$\ be a monomial-like code.
In one variable  case, the minimum Hamming  distance of $C$\ was computed in \cite{ozadam_2009_2} and
\cite{dinh_2008}.
It turns out that, in multivariate case, $C$\ can be considered as a ``product'' of single variable codes.
This decomposition allows us to express the minimum Hamming distance of $C$\ in terms of the Hamming distances
of cyclic codes of length $p^{s_j}$.

Below we define the product of two linear codes.
For the general theory of  product codes, we refer to \cite[Chapter 18]{macwilliams_1977_1}.
\begin{definition}
  \label{Definition.Product.Code}
 The product  of two linear codes $C,C^{'}$\ over $\F_q$\ is the linear code $C \otimes C^{'}$\
 whose codewords are all the two dimensional arrays for which each row is a codeword in $C$\
 and each column is a codeword in $C^{'}$.
\end{definition}

\begin{remark}
  \label{Remark.Properties.of.Product.Codes}
 The following are some well-known facts about the product codes.
 \begin{enumerate}
  \item If $C$\ and $C^{'}$\ are $[n,k,d]$\ and $[n^{'},k^{'},d^{'}]$\ codes respectively,
    then $C \otimes C^{'}$\ is a $[nn^{'}, kk^{'},dd^{'}]$\ code.
  \item If $G$\ and $G^{'}$\ are generator matrices of $C$\ and $C^{'}$\ respectively,
    then $G \otimes G^{'}$\ is a generator matrix of $C \otimes C^{'}$, where $\otimes$\ denotes
    the Kronecker product of matrices and the codewords of $C \otimes C^{'}$\ are seen as
    concatenations of the rows in arrays in $C \otimes C^{'}$.
 \end{enumerate}
\end{remark}

First, we prove that, in two variable case,
monomial-like codes are product codes.

\begin{theorem}
 \label{Theorem.Monomial.Codes.Are.Product.Codes}
  Let $n_1,n_2$\ be positive integers and let
  \be
    \hat{\sR} = \frac{\F_q[x,y]}{ \id{x^{n_1}-1,y^{n_2}-1} },\nn\\
    \sR_x = \frac{\F_q[x]}{ \id{x^{n_1}-1} },\quad \sR_y = \frac{\F_q[y]}{\id{y^{n_2}-1}}. \nn
  \ee
  Suppose that $(x-1)^{k_1} | x^{n_1} -1$\ and $(y-1)^{k_2} | y^{n_2} -1$.
  The code $C = \id{(x-1)^{k_1}(y-1)^{k_2}} \subset \hat{\sR}$\ is the \textit{``product''} of the codes
  $C_x = \id{(x-1)^{k_1}} \subset \sR_x$\ and $C_y = \id{(y-1)^{k_2}} \subset \sR_y$, i.e.,
  $ C = C_x \otimes C_y $.
\end{theorem}
\begin{proof}
 Let
 \be
  g(x) & = & (x-1)^{k_1} = g_{k_1}x^{k_1} + \cdots + g_{1}x + g_0, \nn\\
  h(y) & = & (y-1)^{k_2} = h_{k_2}y^{k_2} + \cdots + h_{1}y + h_0.  \nn
 \ee
 Then\\
 \be G_x & = & \left[ \begin{array}{rrrrrrrr}
     0 & \dots & 0 & 0 & g_{k_1} & \dots & g_1 & g_0\\
     0 & \dots & 0 & g_{k_1} & \dots & g_1 & g_0 & 0\\
     \vdots & & & & & & & \vdots\\
     g_{k_1} & \dots & g_1 & g_0 & 0 & \dots & 0 & 0\\
  \end{array}\right],\nn\\
  G_y & = & \left[ \begin{array}{rrrrrrrr}
     0 & \dots & 0 & 0 & h_{k_2} & \dots & h_1 & h_0\\
     0 & \dots & 0 & h_{k_2} & \dots & h_1 & h_0 & 0\\
     \vdots & & & & & & & \vdots\\
     h_{k_2} & \dots & h_1 & h_0 & 0 & \dots & 0 & 0\\
  \end{array}\right]\nn
 \ee
 are two generator matrices for $C_x$\ and $C_y$, respectively.
 The Kronecker product of $G_x$\ and $G_y$\ is an $d_1 \cdot d_2 \times n_1\cdot n_2$\ matrix given by
 {\tiny{\be\nn
   G_x \otimes G_y = \\
   \left[
      \begin{array}{rrrrrrrrrrrrrrrr}
       0 & 0 &\dots & 0 & 0 & 0 & g_{k_1}h_{k_2} & g_{k_1}h_{k_2-1} & \cdots & g_{k_{1}}h_0 & g_{k_1 -1}h_{k_2} & \dots & g_0h_{k_2} & \dots & g_0h_1 & g_0h_0\\
       0 & \dots & 0 & g_{k_1}h_{k_2} & g_{k_1}h_{k_2-1} & \cdots & g_{k_{1}}h_0 & g_{k_1 -1}h_{k_2} & \dots & g_0h_{k_2} & \dots & g_0h_1 & g_0h_0 & 0 & \cdots & 0\\
       %rrrrrrrrrrrrrrr
       \vdots &&&&&&&&&&&&&&&\vdots\\
       g_{k_1}h_{k_2} & g_{k_1}h_{k_2-1} & \cdots & g_{k_{1}}h_0 & g_{k_1 -1}h_{k_2} & \dots & g_0h_{k_2} & \dots & g_0h_1 & g_0h_0 & 0 & 0 & \cdots  & 0 & 0 & 0\\
      \end{array}
   \right].&&\nn
 \ee}}

 Next, for a polynomial
 \be\nn
  f(x,y) = \sum_{
    \begin{array}{c}
      0 \le  i <  n_1\\
      0 \le  j <  n_2\\
    \end{array}
  }c_{ij}x^iy^j \in \F_q[x,y],
 \ee
 we use the monomial ordering $x>y$\ to order its terms.
 According to this ordering, we identify $f(x,y)$\ with the tuple
 $(c_{n_1-1,n_2-1},c_{n_1-1,n_2-2},\dots , c_{n_1,0}, \dots, c_{n_1-2,n_2-1},\dots,c_{n_1-2,0},\dots,c_{0,0})$.
 Since the elements of $C = \id{(x-1)^{k_1}(y-1)^{k_2}} \subset \hat{\sR}$\ are exactly all the $\F_q$-linear combinations of the
 elements of the set
 \be\nn
  \beta = \{ x^iy^j(x-1)^{k_1}(y-1)^{k_2}: \quad 0 \le i < n - k_1,\quad 0 \le j < n- k_2 \},
 \ee
 we obtain a generator matrix for $C$\ as
 {\tiny{\be\nn
   G = \\
   \left[
      \begin{array}{rrrrrrrrrrrrrrrr}
       0 & 0 &\dots & 0 & 0 & 0 & g_{k_1}h_{k_2} & g_{k_1}h_{k_2-1} & \cdots & g_{k_{1}}h_0 & g_{k_1 -1}h_{k_2} & \dots & g_0h_{k_2} & \dots & g_0h_1 & g_0h_0\\
       0 & \dots & 0 & g_{k_1}h_{k_2} & g_{k_1}h_{k_2-1} & \cdots & g_{k_{1}}h_0 & g_{k_1 -1}h_{k_2} & \dots & g_0h_{k_2} & \dots & g_0h_1 & g_0h_0 & 0 & \cdots & 0\\
       %rrrrrrrrrrrrrrr
       \vdots &&&&&&&&&&&&&&&\vdots\\
       g_{k_1}h_{k_2} & g_{k_1}h_{k_2-1} & \cdots & g_{k_{1}}h_0 & g_{k_1 -1}h_{k_2} & \dots & g_0h_{k_2} & \dots & g_0h_1 & g_0h_0 & 0 & 0 & \cdots  & 0 & 0 & 0\\
      \end{array}
   \right].&&\nn
 \ee}}
 It is easily seen that $G_x \otimes G_y = G$.
 Note that in the above construction, we multiplied $(x-1)^{k_1}(y-1)^{k_2}$\ with the monomials in the order\\
  $x^{0}y^{0},x^{0}y^{1},x^{0}y^{2},\dots , x^{0}y^{n-k_2-1},x^{1}y^{0},\dots , x^{1}y^{n-k_2-1},\dots,
    x^{n-k_1-1}y^{0},x^{n-k_1-1}y^{1},x^{n-k_1-1}y^{n-k_2-1}$\ and considered the corresponding tuples and
  placed these tuples into $G$\ in that order.
\end{proof}

Using the arguments in the proof of Theorem \ref{Theorem.Monomial.Codes.Are.Product.Codes} inductively,
it is straightforward to generalize Theorem \ref{Theorem.Monomial.Codes.Are.Product.Codes} to the multivariate case.

\begin{theorem}
 \label{Theorem.Generalisation.Monomial.Codes.Are.Product.Codes}
  Let $r_1, \dots , r_n,i_1,\dots , i_n$\ be positive integers and let
  \be
    \sR^{'} = \frac{\F_q[x_1, \dots, x_n]}{ \id{x_1^{r_1}-1, \dots, x_n^{r_n}-1} },\quad
    \sR_{x_j} = \frac{\F_q[x_j]}{ \id{x_j^{r_j}-1} }. \nn
  \ee
  Suppose that $(x_j-1)^{i_j} | x_j^{r_j} -1$\ for all $1 \le j \le n$.
  The code 
  \be
     C = \id{ (x_1 - 1)^{i_1} \cdots (x_r - 1)^{i_r} }\nn
  \ee
   is the \textit{``product''} of the codes
  $C_{x_j} = \id{ (x_j -1)^{i_j} } \subset \sR_{x_j}$, i.e.,
  \be
    C = ( \cdots (( C_{x_1} \otimes C_{x_2})\otimes C_{x_3}) \otimes \cdots ) \otimes C_{x_n}.\nn
  \ee
\end{theorem}

Using Theorem \ref{Theorem.Generalisation.Monomial.Codes.Are.Product.Codes}, we determine
the minimum Hamming distance of monomial-like codes.
\begin{theorem}
 \label{Theorem.Minimum.Hamming.Weight.of.Monimial-Like.Codes}
 Let $C = \id{ (x_1 - 1)^{i_1} \cdots (x_n - 1)^{i_n} }\subset \sR$.
 Let $\sR_{x_j} = \frac{\F_q[x_j]}{\id{x^{p^{s_j}} -1}}$\ and $C_{x_j} = \id{ (x_j - 1)^{i_j} } \subset R_{x_j}$.
 Then $d_H(C) = \prod_{j=1}^{n}d_H(C_{x_j})$, where
 $d_H(C_{x_j})$\ is as given in \cite[Theorem 6.4]{dinh_2008} and \cite[Theorem 1]{ozadam_2009_2}.
\end{theorem}
\begin{proof}
 We have
  $C = ( \cdots (( C_{x_1} \otimes C_{x_2})\otimes C_{x_3}) \otimes \cdots ) \otimes C_{x_n}$\
  by Theorem \ref{Theorem.Generalisation.Monomial.Codes.Are.Product.Codes}.
  The result follows by Remark \ref{Remark.Properties.of.Product.Codes} (2).
\end{proof}

\subsection{The dual of monomial-like codes}\ \\

We determine the dual of
\be
  C = \id{ (x_1 - 1)^{N_1} \cdots (x_n - 1)^{N_n} } \subset 
      \frac{\F_q[x_1, \dots , x_n]}{ \id{x_1^{p^{s_1}} - 1 , \cdots , x_n^{p^{s_n}} - 1  } } = \sR. \nn
\ee
Let $L\subset \N^n$\ and $i=(i_1,\dots , i_n)\in L$.
We consider $\fm[x] \in \sR$\ in the form
\be
  \fm[n] = \sum_{i \in L}c_i(x_1 - 1)^{i_1}\cdots (x_n - 1)^{i_n}\nn
\ee
where $c_i \neq 0 $, for all $i \in L$.
We define
\be
  L_1 & = & \{ (i_1, \dots , i_n): \quad i_j < p^{s_j} - N_j\quad \forall \ 1\le j \le n \},\nn\\
  L_2 & = & \{ (i_1, \dots , i_n): \quad i_j \ge p^{s_j} - N_j\quad \mbox{for some } \quad 1 \le j \le n \}.\nn
\ee
This gives us a partition of $L$\ as $L = L_1 \sqcup L_2$. Since
\be
  (x_1 - 1)^{N_1} \cdots (x_n - 1)^{N_n}\fm[n] = 
    (x_1 - 1)^{N_1} \cdots (x_n - 1)^{N_n}\sum_{i \in L_1}c_i(x_1 - 1)^{i_1}\cdots (x_n - 1)^{i_n},\nn
\ee
 we deduce that
$(x_1 - 1)^{N_1} \cdots (x_n - 1)^{N_n}\fm[n] = 0$\ if and only if
$(x_1 - 1)^{N_1} \cdots (x_n - 1)^{N_n}\sum_{i \in L_1}c_i(x_1 - 1)^{i_1}\cdots (x_n - 1)^{i_n} = 0$\ if and only if
$(x_j - 1)^{p^{s_j} - N_j} | \fm[x]$\ for some $1 \le j \le n$.
Equivalently, if $\fm[n] \in C^{\perp}$, then $\fm[x] \in \id{ (x_1 - 1)^{p^{s_1} - N_1}, \dots , (x_n - 1)^{p^{s_n} - N_n} }$.
Conversely, $(x_j - 1)^{p^{s_j} - N_j}(x_1 - 1)^{N_1} \cdots (x_n - 1)^{N_n} = 0$\ for all $1 \le j \le n$.
This proves the following.

\begin{lemma}
  \label{Lemma.Dual.of.Monomial.like.Code}
 Let
 \be
  C = \id{ (x_1 - 1)^{N_1} \cdots (x_n - 1)^{N_n} } \subset 
      \frac{\F_q[x_1, \dots , x_n]}{ \id{x_1^{p^{s_1}} - 1 , \cdots , x_n^{p^{s_n}} - 1  } }. \nn
  \ee
  Then
  \be
    C^{\perp} = \id{ (x_1 - 1)^{p^{s_1} - N_1}, \dots , (x_n - 1)^{p^{s_n} - N_n} }.\nn
  \ee
\end{lemma}
\begin{remark}
 Lemma \ref{Lemma.Dual.of.Monomial.like.Code} does not hold for arbitrary codeword lengths.
 That is, if
 \be
  \sR^{'} = \frac{\F_q[x_1, \dots, x_n]}{ \id{x_1^{A_1} -1, \dots , x_n^{A_n} - 1}  },\nn
 \ee
 then since $x_j^{A_j} - 1 = (x-1)^{A_j}$\ only when $A_j = p^{s_j}$\ for some
 $s_j$, the above arguments are not valid for arbitrary $A_j$.
\end{remark}

Via Lemma \ref{Preliminaries.Isomorphism}, Lemma \ref{Lemma.Dual.of.Monomial.like.Code} can be generalized to
constacyclic codes.

\begin{lemma}
  \label{Lemma.Dual.of.Monomial.like.Code.Constacyclic.Generalisation}
 Let
 \be
  D = \id{ (x_1 - c_1)^{N_1} \cdots (x_n - c_n)^{N_n} } \subset 
      \frac{\F_q[x_1, \dots , x_n]}{ \id{x_1^{p^{s_1}} - c_1^{p^{s_1}} , \cdots , x_n^{p^{s_n}} - c_n^{p^{s_n}}  } }. \nn
  \ee
  Then
  \be
    D^{\perp} = \id{ (x_1 - c_1)^{p^{s_1} - N_1}, \dots , (x_n - c_n)^{p^{s_n} - N_n} }.\nn
  \ee
\end{lemma}
Now we construct an $\F_q$-basis for $C^{\perp}$.
This also gives us a generator matrix for $C^{\perp}$\ and hence a parity check matrix for $C$.

 We define
 \be
  T & = & \{ (a_1,\dots, a_n) \in \N^n: \quad p^{s_j} - N_j \le a_j < p^{s_j} \}\quad \mbox{and} \nn\\
  B & = & \{ (x_1-1)^{a_1}\cdots (x_n - 1)^{a_n}: \quad (a_1, \dots , a_n)\in T \}.\nn
 \ee
 Since the set
 $ B^{'} = \{ x_1^{a_1}\cdots x_n^{a_n} : \quad (a_1,\dots, a_n) \in T \} $
 is linearly independent, by the isomorphism given in Remark \ref{Remark.Isomorhpism.Basic.Change.of.Variable}, 
 we see that the set $B$\ is linearly independent. 
Let
\be
  T_j = \{ (a_1, \dots , a_n) \in \N^n:\quad p^{s_j}- N_j \le a_j < p^{s_j}  \}.\nn
\ee
then we can view $T$\ as
\be
  \label{Equation.T.Decomposition}
  T = T_1 \cup \cdots \cup T_n.
\ee
Let $p^s = p^{s_1}\cdots p^{s_n}$. Note that $|T_j| = N_j\frac{p^s}{p^{s_j}}$\ moreover,
\be
  |T_{e_1} \cap \dots \cap T_{e_r}| = N_{e_1}\cdots N_{e_r}\frac{p^s}{p^{s_{e_1}} \cdots p^{s_{e_r}} }.\nn
\ee
Now applying the inclusion-exclusion principle to (\ref{Equation.T.Decomposition}), we obtain
\be
  |T| & = & N_1^{p^s / p^{s_1}} + \cdots + N_n^{p^s / p^{s_n}} - N_1 N_2\frac{p^s}{p^{s_1}p^{s_2}}- \cdots 
	  - N_{n-1} N_n\frac{p^s}{p^{s_{n-1}}p^{s_n}} + \cdots + (-1)^nN_1\cdots N_n  \nn\\
	& = & p^{s_1}\cdots p^{sNn} - (p^{s_1} - N_1)\cdots (p^{s_n} -  N_n).\nn  
\ee
Clearly $|B| = |T| = p^{s_1}\cdots p^{s_n} - (p^{s_1} - N_1)\cdots (p^{s_n} -  N_n)$.
On the other hand, we know, from Theorem \ref{Theorem.Monomial.Codes.Are.Product.Codes}, that
$\dim (C) =  (p^{s_1} - N_1)\cdots (p^{s_n} -  N_n) $. 
This implies that
$\dim (C) = p^{s_1}\cdots p^{s_n} - \dim(C) = p^{s_1}\cdots p^{s_n} - (p^{s_1} - N_1)\cdots (p^{s_n} -  N_n)$.
Therefore the set $B$\ is an $\F_q$-basis for $C^{\perp}$.
Considering the vector representations of the elements of $B$, we obtain a generator matrix
for $C^{\perp}$\ and a parity check matrix for $C$.
In Section \ref{Section.Hasse.Der}, we present another method of finding a parity check matrix for $C$.

In particular, in 2 variable case, there are few enough cases to express $B$\ and $T$ more explicitly
in a feasible way.

We define 
 \be
  T^{(2)} & = & \{ (k,m): \quad p^{s_1 - i} \le k < p^{s_1} \quad \mbox{and} \quad p^{s_2 - j} \le m  < p^{s_2} \}  \nn \\
	&& \sqcup \{ (k,m): \quad 0 \le k < p^{s_1 - i} \quad \mbox{and} \quad p^{s_2 - j} \le m < p^{s_2} \}\nn \\
	&& \sqcup \{ (k,m): \quad p^{s_1 - i} \le k < p^{s_1} \quad \mbox{and} \quad o \le m < p^{s_2 - j}  \}.\nn
 \ee
 The set
 $$ B^{(2)} = \{ (x-1)^k(y-1)^m: \quad (k,m) \in T^{(2)} \} $$
 is linearly independent and $|T^{(2)}| = p^{s_1}p^{s_2} - (p^{s_1} - i)(p^{s_2} - j)$.
 Hence $B^{(2)}$ is an $\F_q$-basis for $C_2^{\perp}$.
 
%%%%%%%%%%%%%%%%%%%%%%%%%%%%%%%%%%%%%%%%%%%%%%%%%%%%%%%%%%%%%%%%%%%%%%%%%%%%%%%%%%%%%%%%%%%%
%%%%%%%%%%%%%%%%% G E N E R A L I Z E D    H A M M I N G    W E I G H T   %%%%%%%%%%%%%%%%%%
%%%%%%%%%              O F    M O N O M I A L - L I K E     C O D E S           %%%%%%%%%%%%
%%%%%%%%%%%%%%%%%%%%%%%%%%%%%%%%%%%%%%%%%%%%%%%%%%%%%%%%%%%%%%%%%%%%%%%%%%%%%%%%%%%%%%%%%%%%
\section{Weight hierarchy of some monomial-like codes}
\label{Section.Generalized.Ham.W.Prod.Codes}

Let
\be
  \sR^{'} = \frac{\F_q[x]}{\id{x^p -1}}\nn
\ee
and $C = \id{(x-1)^i} \subset \sR^{'}$.
It was shown in \cite[Theorem 5]{MASSEY_1973}
that $C$\ is an MDS code. The weight hierarchy of MDS codes are determined in \cite[Theorem 7.10.7]{huffman_2003}.
First we state the weight hierarchy of $C$\ and prove it for the sake completeness.
Next we study the weight hierarchy of the product of two monomial-like codes which are subsets of $\sR^{'}$.
For such codes, we simplify an expression that gives us the weight hierarchy of monomial-like codes of the form
\be
  \label{Equation.Definition.C.xy}
  C_{xy} = \id{(x-1)^i(y-1)^j} \subset \frac{\F_q[x,y]}{ \id{x^p-1, y^p-1} }.
\ee
We begin by giving the necessary definitions and facts.
The reader is referred to \cite[Section 7.10]{huffman_2003} or \cite{martinez_perez_2001} for the details.
The support of a codeword $c = (c_1, \dots, c_m)$\ is the set
\be
  \chi(c) = \{ i: \quad c_i \neq 0 \}. \nn
\ee
The support of a subset $S \subset \F_q^m$\ is the set
\be\nn
  \chi(S) = \bigcup_{c \in S} \chi(c).
\ee
If $D \subset \F_q^m$\ is a subspace of $C$, then we denote this by $D \leq C$.
The $r^{th}$\ minimum Hamming weight of a code $C$\ is defined as
\be
  d_r(C) = \min \{ \#\chi(D): D \leq C, \quad \dim (D) = r\}.\nn
\ee
The weight hierarchy of a k-dimensional code $C$\ is the sequence
\be
   (d_1(C),d_2(C), \dots , d_k(C)). \nn
\ee
An easy, yet important, observation is that $d_1(C) = d_H(C)$.
To see this, consider the 1-dim  subspace of $C$\ generated by a minimum
weight codeword.

Now we give a lower bound on the generalized Hamming weight of $C$.
\begin{lemma}
  \label{Gen.HW.Prod.Codes.Lemma.Chi.D.p.Case.Lower.Bound}
    Let $D \subset C = \id{(x-1)^i} \subset \frac{\F_q[x]}{\id{x^p - 1}}$\ be
    a k-dimensional subspace of $C$. Then
    \be
      \chi(D) > i + k - 1. \nn
    \ee
\end{lemma}
\begin{proof}
 Assume the converse. Let $ \sB = \{ \beta_1,\dots , \beta_k \}$\ be a basis for $D$.
 For
  \be
    \beta_1 = (\beta_{1,1}, \dots , \beta_{1,p}),\dots , \beta_k = (\beta_{k,1}, \dots , \beta_{k,p}),\nn
   \ee
 let $e_1, \dots , e_{i+k-1}$ be the coordinates  where $\beta_1, \dots , \beta_k$\ are possibly nonzero.
 In other words, the generators $\beta_1, \dots , \beta_k$\ (hence all the elements of $D$)\ are zero
 at the coordinates $\{1,2, \dots , p\}\setminus \{ e_1, \dots , e_{i+k-1} \}$.
 For $ 1 \le \ell \le k $, define
 \be
  \beta_{\ell}^{'} = (\beta_{\ell,e_1}, \dots , \beta_{\ell,e_{i + k - 1}}).\nn
 \ee
 Since $\beta_1, \dots , \beta_k$\ are linearly independent, the vectors
 $\beta_1^{'}, \dots , \beta_k^{'}$\ are also linearly independent.
 Then, after some rearrangement if necessary, applying Gaussian elimination, we can
 put these vectors in such a form, say $\alpha_1, \dots , \alpha_k$, that
 each $\alpha_{\ell}$\ has at least $N+(\ell-1)$\ leading zeroes where $N \ge 0$.
 Thus, the vector $\alpha_k$\ has at least $k-1$\ leading zeroes.
 So $w_H(\alpha_k) \le i + k -1 - (k-1) = i$.
 This implies that there is a codeword $\hat{\alpha}_k$,
 which is obtained after putting back the stripped off zeroes, with $d_H(\hat{\alpha}_k) < i+1$.
 This is a contradiction because $d_H(C) = i+1$. Hence $\chi (D) > i + k -1$.
\end{proof}

Using the above lower bound, we determine the generalized Hamming weight of $C$.
We would like to note that the next corollary is an immediate consequence of 
\cite[Theorem 7.10.7]{huffman_2003}.
\begin{corollary}
 \label{Gen.HW.Prod.Codes.Corollary.Chi.D.p.Case.Exact.Value}
 Let $C = \id{(x-1)^i} \subset \frac{\F_q[x]}{\id{x^p - 1}}$. Then
 \be
  d_r(C) = i + r.\nn
 \ee
\end{corollary}
\begin{proof}
 Since
 \be
  d_r(C) = \min \{ \# \chi(D): \quad D \le C, \quad \dim (D) = r \}, \nn
 \ee
 by Lemma \ref{Gen.HW.Prod.Codes.Lemma.Chi.D.p.Case.Lower.Bound},
 it suffices to show that there exists an $r$\ dimensional subspace
 $D$, of $C$\ such that $\chi (D) = i + r$.
 Consider the subspace $T = \id {(x-1)^i, x(x-1)^i, \dots , x^{r-1}(x-1)^i}$.
 Obviously, the generators are linearly independent and $\dim (T) = r$.
 It is not hard to see that $\chi(D) = i+1 + (r-1) = i+r$.
 This completes the proof.
\end{proof}

\begin{remark}
  \label{Gen.HW.Prod.Codes.Remark.Weight.Hier.p.Case}
  Corollary \ref{Gen.HW.Prod.Codes.Corollary.Chi.D.p.Case.Exact.Value} gives us the weight hierarchy of
  all cyclic codes of length $p$\ over a finite field of characteristic $p$.
  With the notation in Corollary \ref{Gen.HW.Prod.Codes.Corollary.Chi.D.p.Case.Exact.Value},
  the weight hierarchy of $C$\ is
  \be
    \label{Equation.Weight.Hierarchy.of.C}
    (d_1(C), d_2(C),\dots , d_{p-i}(C) ) = ( i+1, i+2, \dots, p ).
  \ee
\end{remark}

Using the weight hierarchy of $C$\ (\ref{Equation.Weight.Hierarchy.of.C}),
we study the weight hierarchy of codes that are product of cyclic codes 
of length $p$\ over $\F_q$.

A $(k_1,k_2)$-partition of an integer  $r$\ is a non-increasing sequence
$\pi = (t_1, \dots, t_{k_1})$\ such that $t_1 + \cdots + t_{k_1} = r$\ and
$t_i \le k_2$\ for all $1 \le i \le k_1$. We denote all the $(k_1,k_2)$\
partitions of $r$\ by $P(k_1,k_2,r)$.

Let $D_1,D_2$\ be $[n_1,k_1,d_1],[n_2,k_2,d_2]$\ linear codes, respectively.
Let
\be
  \bigtriangledown (\pi) =  \sum_{i=1}^{k_1}(d_i(D_1) - d_{i-1}(D_1))d_{t_i}(D_2),\quad \pi \in P(k_1,k_2,r),
    \label{Gen.HW.Prod.Codes.Equation.Equation.Delta.Pi}\\
  d_r^*(D_1 \otimes D_2) = \min \{ \bigtriangledown (\pi): \quad \pi \in P(k_1,k_2,r) \}.\nn
\ee 
From \cite[Theorem 1]{schaathun_2000} , 
we know that $d_r^*(D_1 \otimes D_2) = d_r(D_1 \otimes D_2)$.

Now, for $C_1 = \id{(x-1)^{i_1}} \subset \frac{\F_q[x]}{\id{x^p - 1}}$\ and
 $C_2 = \id{(y-1)^{i_2}} \subset \frac{\F_q[y]}{\id{y^p - 1}}$,
using (\ref{Equation.Weight.Hierarchy.of.C}),
the expression (\ref{Gen.HW.Prod.Codes.Equation.Equation.Delta.Pi}) simplifies to
\be
  \bigtriangledown (\pi) & = & (i_1 + 1)d_{t_1}(C_2) + \sum_{i = 2}^{k_2}d_{t_i}(C_2)\nn\\
            & = & (i_1 + 1)(i_2 + t_1) + \sum_{i = 2}^{k_2}(i_2 + t_i).
            \label{Gen.HW.Prod.Codes.Equation.Delta.Pi.Simplified.p.Case}
\ee

In the following lemmas, we consider the cyclic codes $C_1, C_2$\ which are as introduced above
with the same notation.

\begin{lemma}
 \label{Gen.HW.Prod.Codes.Lemma.Minimum.Delta.Partition}
 Let $\pi_0 = (m,t_2, \dots , t_k)$\ be a $(k_1,k_2)$-partition
 of $r$\ such that $\bigtriangledown (\pi_0)$\ is minimum among all
 $\bigtriangledown (\hat{\pi})$\ where $\hat{\pi} \in P(k_1,k_2,r)$.
 Then, for the $(k_1,k_2)$-partition $\pi = (m,m,\dots , m,u,0, \dots , 0)$\ of
 $r$, where $0 < u \le m$, we have
 \be
  \bigtriangledown (\pi_0) = \bigtriangledown (\pi). \nn
 \ee
\end{lemma}
\begin{proof}
 Say $\pi = (a_1, \dots , a_{e-1}, a_e, a_{e+1}, \dots , a_{k_1})$, where
 $a_1 = \cdots = a_{e-1} = m$, $a_e = u$\ and $a_{e+1} = \cdots = a_{k_1} = 0$.
 If $\pi = \pi_0$, then we are done.
 If $a_i = t_i$\ for all $1 \le i \le e-1$\ and $a_e = u \ge t_e$, then since
 $\sum_{i=1}^{k_1}a_i = \sum_{i=1}^{k_i}t_i$, we get
 $u = t_e + t_{e+1} + \cdots + t_{e + \ell}$\ for some $\ell \ge 0$.
 Now, by (\ref{Gen.HW.Prod.Codes.Equation.Delta.Pi.Simplified.p.Case}), we get
 \be
  \bigtriangledown (\pi) & = & (i_1 + 1)(i_2 + m) + (e-2)(i_2 + m) + (i_2 + u), \quad \mbox{and}\nn\\
  \bigtriangledown (\pi_0) & = & (i_1 + 1)(i_2 + m) + (e-2)(i_2 + m) + \sum_{j=0}^{\ell}(i_2 + t_{e + j}). \nn
 \ee
 So, by the minimality of $\bigtriangledown (\pi_0)$, we get
 $\bigtriangledown (\pi) - \bigtriangledown (\pi_0) = i_2 + u - \sum_{j=0}^{\ell}(i_2 + t_{e + j}) \ge 0$.
 This implies $\ell = 0 $\ and $t_e = u$, $t_{e+1} = \cdots = t_{k_1} = 0$.
 Hence, in this case, $\pi = \pi_0$.
 If $a_{\alpha} < t_{\alpha} = m $\ for some $1 < \alpha \le e-1$, then,
 since $\pi_0$\ is a  non-increasing sequence, $\pi_0$\ is of the form
 \be
  \pi_0 = (m,m,\dots , m , t_{\alpha},t_{\alpha + 1}, \dots , t_N, 0 , \dots , 0) \nn
 \ee
 for some $\alpha + 1 \le N \le k_1$, where $t_j < m$\ for all $j \ge \alpha$.
 This implies that $N \ge e$. So
 \be
  \bigtriangledown (\pi_0) = (i_1 + 1)(i_2 + m) + (\alpha -2)(i_2 + m) + \sum_{j = \alpha}^{N}(i_2 + t_i). \nn
 \ee
 On the other hand,
 \be
  \bigtriangledown (\pi) = (i_1 + 1)(i_2 + m) + (\alpha -2)(i_2 + m) + \sum_{j = \alpha}^{e}(i_2 + a_i). \nn
 \ee
 Since $\sum_{j = \alpha}^N t_i = \sum_{j = \alpha}^e a_i$, we get
 \be
  \bigtriangledown (\pi) - \bigtriangledown (\pi_0) = (e - \alpha + 1)i_2 - (N - \alpha + 1)i_2 \ge 0, \nn
 \ee
 by the minimality of $\bigtriangledown (\pi_0)$. By the fact that $N \ge e$, we obtain
 \be
  \bigtriangledown (\pi) - \bigtriangledown (\pi_0) = (e - \alpha )i_2 - (N - \alpha)i_2 \le 0. \nn
 \ee
 Thus $\bigtriangledown (\pi) = \bigtriangledown (\pi_0)$.
\end{proof}

\begin{lemma}
 \label{Gen.HW.Prod.Codes.Lemma.Simplified.dr}
 Let $r$\ be an integer such that $\alpha k_1 < r \le (\alpha + 1)k_1$. Let
 \be
  S = \{ (\beta, \beta , \dots , \beta, u_{\beta}, 0 , \dots , 0) \in P(k_1,k_2,r): \quad \alpha + 1 \le \beta \le \min \{k_2,r\} \}.\nn
 \ee
 We  have
 \be
  d_r(C_1 \otimes C_2) = d_r^{*}(C_1 \otimes C_2) = d_r(C_1 \otimes C_2) = \min \{ \bigtriangledown(\pi): \quad \pi \in S \}.\nn
 \ee
\end{lemma}

Lemma \ref{Gen.HW.Prod.Codes.Lemma.Simplified.dr} simplifies the computation of $d_r^{*}(C_1 \otimes C_2)$
significantly. The search set, for the minimum of $\bigtriangledown(\pi)$, reduces from the set of all
$(k_1,k_2)$ partitions of $r$\ to the set of $(k_1,k_2)$\ partitions of $r$\ that are of the form
$(\beta, \beta , \dots , \beta, u_{\beta}, 0 , \dots , 0)$.

Let $C_{xy}$\ be as in (\ref{Equation.Definition.C.xy}).
We know that $C_{xy} = C_1 \otimes C_2$\ by Theorem \ref{Theorem.Monomial.Codes.Are.Product.Codes}.
Therefore, the above simplification also applies to the generalized Hamming weight
of the monomial-like code $C_{xy}$. More explicitly, we have shown that 
\be\nn
  d_r(C_{xy}) =  \min \{ \bigtriangledown(\pi): \quad \pi = (\beta, \beta , \dots , \beta, u_{\beta}, 0 , \dots , 0),
    \quad \pi \in P(k_1,k_2,r) \}.
\ee

%%%%%%%%%%%%%%%%%%%%%%%%%%%%%%%%%%%%%%%%%%%%%%%%%%%%%%%%%%%%%%%%%%%%%%%%%%%%%%%%%%%
%%%%%%%%  H A S S E     D E R I V A T I V E    A N D %%%%%%%%%%%%%%%%%%%%%%%%%%%%%%
%%%%%%%% P A R I T Y    C H E C K     M A T R I X     C O N S T R U C T I O N %%%%%
%%%%%%%%%%%%%%%%%%%%%%%%%%%%%%%%%%%%%%%%%%%%%%%%%%%%%%%%%%%%%%%%%%%%%%%%%%%%%%%%%%%
\section{Construction of parity check matrix and the Hasse derivative}
\label{Section.Hasse.Der}

We begin by recalling the Hasse derivative which is used in the repeated-root factor test.
For a detailed treatment of the Hasse derivative, we refer to
\cite[Chapter 1]{Goldschmidt_book} and \cite[Chapter 5]{Torres_Algeb_Curves}.

The standard derivative for polynomials over a field of positive characteristic, say $p$,
is inappropriate because from the $p^{th}$ derivative on, the result is always zero.
For this reason, it is more convenient to work with the Hasse derivative.
Sometimes the Hasse derivative is called as the hyper derivative.

Throughout this section, we will use the convention that ${a \choose b} = 0$\ whenever
$b < a$.
Let $\gm = \sum d_{i_1,\dots , i_n}x_1^{i_1}\cdots x_n^{i_n}\in \F_q[x_1, \dots , x_n]$.  The classical derivative of $\gm$\
in the direction $(a_1, \dots, a_n)$\ is defined as
\be
  \frac{\partial^{(a_1+ \dots + a_n)}}{\partial x_1^{a_1} \dots \partial x_n^{a_n}}\gm
    = \sum d_{i_1,\dots , i_n}a_1!\cdots a_n!{i_1 \choose a_1}\cdots {i_n \choose a_n} x_1^{i_1-a_1}\cdots x_n^{i_n-a_n}.\nn
\ee
The Hasse derivative of $\gm$\ in the direction  $(a_1, \dots, a_n)$\ is defined as
\be
  D^{[a_1, \dots , a_n]}(\gm) = \sum d_{i_1,\dots , i_n}{i_1 \choose a_1}\cdots {i_n \choose a_n} x_1^{i_1-a_1}\cdots x_n^{i_n-a_n}.\nn
\ee
We denote the evaluation of $D^{[a_1, \dots , a_n]}(\gm)$\ at the point $(\alpha_1, \dots , \alpha_n)$\ by
$D^{[a_1, \dots , a_n]}(g)(\alpha_1, \dots , \alpha_n)$.
We can  express $\gm$\ as
\be
   \gm = \sum_{(i_1,\dots , i_n) \in S} c_{i_1, \dots , i_n}(x_1 - 1)^{i_1}\cdots (x_n - 1)^{i_n} \nn
\ee
where $S$\ is a finite nonempty subset of $\N^n$. Let
\be
  U_{\ell} & = & \{ (i_1, \dots , i_n) \in S: \quad i_{\ell} \ge m_{\ell} \},\nn\\
  P_{\ell} & = & \{ (i_1, \dots , i_n) \in S: \quad i_{\ell} < m_{\ell} \}.\nn
\ee
Obviously $S = U_{\ell} \sqcup P_{\ell}$. So
\be
  \gm & = & \sum_{(i_1,\dots , i_n) \in U_{\ell}} c_{i_1, \dots , i_n}(x_1 - 1)^{i_1}\cdots (x_n - 1)^{i_n}\nn\\ 
	  &&+ \sum_{(i_1,\dots , i_n) \in P_{\ell}} c_{i_1, \dots , i_n}(x_1 - 1)^{i_1}\cdots (x_n - 1)^{i_n}.\nn
\ee
The term $(x_{\ell} - 1)^{m_{\ell}}$\ divides $\gm$\ if and only if 
$c_{i_1, \dots, i_n} = 0$\ for all $(i_1,\dots, i_n)\in P_{\ell}$.
Now suppose that  $(x_{\ell} - 1)^{m_{\ell}} \nmid \gm$.
Then there is $(\hat{i}_1, \dots , \hat{i}_n) \in P_{\ell}$\ such that
$c_{\hat{i}_1, \dots , \hat{i}_n} \neq 0$. So
\be
  D^{[\hat{i}_1, \dots , \hat{i}_n]}(g)(1,\dots, 1) = 
    c_{\hat{i}_1, \dots , \hat{i}_n} {\hat{i}_1 \choose \hat{i}_1} \cdots {\hat{i}_n \choose \hat{i}_n} \neq 0.\nn
\ee
Conversely, if $(x_{\ell} - 1)^{m_{\ell}}$\ divides \gm , then
\be
  \gm & = & \sum_{(i_1,\dots , i_n) \in U_{\ell}} c_{i_1, \dots , i_n}(x_1 - 1)^{i_1}\cdots (x_n - 1)^{i_n}.\nn
\ee
So $D^{[\vec{a}]}(g)(1, \dots, 1) = 0$\ for all $\vec{a} = (a_1, \dots, a_n)$\ with $0 \le a_{\ell} < m_{\ell}$.
This proves the following.
\begin{lemma}
 Let $\gm \in \F_q[x_1, \dots, x_n]$ and let 
 $A_{\ell} = \{\vec{a} = (a_1, \dots, a_n)\in \N^n: \quad  0 \le a_{\ell} < m_{\ell} \}$.
 Then $(x_{\ell} -1)^{m_{\ell}}$\ divides $\gm$\ if and only if
 $D^{[\vec{a}]}(g)(1, \dots, 1) = 0$\ for all $\vec{a} \in A_{\ell}$.
\end{lemma}

As an immediate consequence, we have the following.

\begin{theorem}
 \label{Theorem.Generalised.Repeated.Root.Factor.Test}
 Let $ A_{\ell} = \{\vec{a} = (a_1, \dots, a_n)\in \N^n: \quad  0 \le a_{\ell} < m_{\ell} \}$\ and
 $A = \cup_{\ell = 1}^{n}A_{\ell}$. Let $\gm \in \F_q[x_1, \dots , x_n]$. We have
 $(x_1 -1)^{m_1}\cdots (x_n - 1)^{m_n}$\ divides $\gm$\ if and only if
 $D^{[\vec{a}]}(g)(1, \dots, 1) = 0$\ for all $\vec{a} \in A$.
\end{theorem}

Let $\sR$\ be as in (\ref{Definition.Ambient.Space}) and 
let $C = \id{ (x_1 - 1)^{m_1}\cdots (x_n-1)^{m_n} } \subset \sR$.
We know that $\gm \in C$\ if and only if $(x_1 - 1)^{m_1}\cdots (x_n-1)^{m_n}$\ divides
$\gm$. Note that $D^{[a_1, \dots, a_n]}(g)(1,\dots, 1) = 0$\ if $a_{\ell} \ge p^{s_{\ell}}$\ for some $1 \le \ell \le n$.
Together with this fact, Theorem \ref{Theorem.Generalised.Repeated.Root.Factor.Test} implies the following.

\begin{theorem}
 \label{Theorem.Generalised.Repeated.Root.Factor.Test.For.Codewords}
 Let  $C = \id{ (x_1 - 1)^{m_1}\cdots (x_n-1)^{m_n} } \subset \sR$.
 Define $Q_{\ell} = \{\vec{a} = (a_1, \dots, a_n)\in \N^n: \quad  0 \le a_{\ell} < m_{\ell}, \quad
   0 \le a_j < p^{s_j}\quad \mbox{for}\quad j \neq \ell \}$\ and
 $Q = \cup_{\ell = 1}^n Q_{\ell}$.
 Then $\gm \in C$\ if and only if $D^{[\vec{a}]}(g)(1,\dots, 1) = 0$\ for all $\vec{a} \in Q$.
\end{theorem}

Fix a monomial order, take $x_1 > \cdots > x_n$.
Let $\vec{a} = (a_1, \dots , a_n)\in Q$. Consider the vector
$$w_a = \left( {p^{s_1} -1 \choose a_1}\cdots {p^{s_n} -1 \choose a_n}, 
	  {p^{s_1} -1 \choose a_1}\cdots {p^{s_{n-1}} -1 \choose a_{n-1}}{p^{s_n} -2 \choose a_n}, \cdots
	  {0 \choose a_1}, \cdots {0 \choose a_n}\right). $$
For $\gm \in \sR$, let $u_g$\ be the vector representation of the polynomial with respect to the fixed ordering.
Then the dot product of $w_a$\ and $u_g$\ gives us the evaluation of the Hasse derivative of $\gm$\
at $(1, \dots, 1)$\ in the direction $\vec{a}$, i.e.,
$w_a \cdot u_g = D^{[\vec{a}]}(g)(1, \dots, 1)$.
Now let $H$\ be a matrix having rows $w_a$\ where $\vec{a} \in Q$\ and $Q$\ is as in 
Theorem \ref{Theorem.Generalised.Repeated.Root.Factor.Test.For.Codewords}. 
Then $H$\ is a parity check matrix for $C$\ by Theorem \ref{Theorem.Generalised.Repeated.Root.Factor.Test.For.Codewords}.

In particular, when there are two variables, we have the following construction.
Let
\be
  \sR_2 = \frac{\F_q[x,y]}{ \id{x^{p^{s_1}} - 1, y^{p^{s_2}} -1} }\nn
\ee
and let $C_2 = \id{(x-1)^i(y-1)^j} \subset \sR_2$.
Define
\be
  A^{(2)} & = & \{ (k,\ell): \quad 0 \le k < k_1 \quad \mbox{and} \quad 0 \le \ell < k_2 \}\nn\\
      && \sqcup \{ (k,\ell): \quad 0 \le k < k_1  \quad \mbox{and} \quad  k_2 \le \ell < p^{s_2} \}\nn\\
      && \sqcup \{ (k,\ell): \quad k_1 \le k < p^{s_1}  \quad \mbox{and} \quad  0 \le \ell < k_2 \}.\nn
\ee
For 
$$ 
  f(x,y) = \sum_{
    \begin{array}{c}
     0 \le i < p^{s_1}\\
     0 \le j < p^{s_2}
    \end{array}
  } f_{ij}x^iy^j, 
$$
using the lexicographic order $x > y$, we can view $f(x,y)$\ as the vector
$$f = (f_{p^{s_1}-1,p^{s_2}-1},f_{p^{s_1}-1,p^{s_2}-2}, \dots , f_{p^{s_1}-1,0}, \dots , f_{0,0}).$$
So
\be
  D^{[k,\ell]}(f)(1,1) = f \cdot \left( {p^{s_1} - 1 \choose k}{p^{s_2} -1 \choose \ell},
      {p^{s_1} - 1 \choose k}{p^{s_2} -2 \choose \ell}, \dots , {0 \choose k}{0 \choose \ell} \right).\nn
\ee
Hence an $(ip^{s_2}+jp^{s_1}-ij)\times(p^{s_1}p^{s_2})$ matrix whose rows are
$$ \left[ {p^{s_1} - 1 \choose k}{p^{s_2} -1 \choose \ell},
      {p^{s_1} - 1 \choose k}{p^{s_2} -2 \choose \ell}, \dots , {0 \choose k}{0 \choose \ell} \right] $$
where $(k,\ell) \in  A^{(2)}$, is a parity check matrix for $C_2$.


\begin{thebibliography}{10}

\bibitem{castagnoli_1991}
Guy Castagnoli, James~L. Massey, Philipp~A. Schoeller, and Niklaus von Seemann.
\newblock On repeated-root cyclic codes.
\newblock {\em IEEE Trans. Inform. Theory}, 37(2):337--342, 1991.

\bibitem{dinh_2008}
Hai~Q. Dinh.
\newblock On the linear ordering of some classes of negacyclic and cyclic codes
  and their distance distributions.
\newblock {\em Finite Fields Appl.}, 14(1):22--40, 2008.

\bibitem{dinh_2004}
Hai~Quang Dinh and Sergio~R. L{\'o}pez-Permouth.
\newblock Cyclic and negacyclic codes over finite chain rings.
\newblock {\em IEEE Trans. Inform. Theory}, 50(8):1728--1744, 2004.

\bibitem{Lakatos_AAECC}
Vesselin Drensky and Piroska Lakatos.
\newblock Monomial ideals, group algebras and error correcting codes.
\newblock In {\em Applied algebra, algebraic algorithms and error-correcting
  codes ({R}ome, 1988)}, volume 357 of {\em Lecture Notes in Comput. Sci.},
  pages 181--188. Springer, Berlin, 1989.

\bibitem{Goldschmidt_book}
David~M. Goldschmidt.
\newblock {\em Algebraic functions and projective curves}, volume 215 of {\em
  Graduate Texts in Mathematics}.
\newblock Springer-Verlag, New York, 2003.

\bibitem{Helleseth_Gen_Ham_Weight_First}
Tor Helleseth, Torleiv Kl{\o}ve, and Johannes Mykkeltveit.
\newblock The weight distribution of irreducible cyclic codes with block length
  {$n_{1}((q^{l}-1)/N)$}.
\newblock {\em Discrete Math.}, 18(2):179--211, 1977.

\bibitem{Torres_Algeb_Curves}
J.~W.~P. Hirschfeld, G.~Korchm{\'a}ros, and F.~Torres.
\newblock {\em Algebraic curves over a finite field}.
\newblock Princeton Series in Applied Mathematics. Princeton University Press,
  Princeton, NJ, 2008.

\bibitem{huffman_2003}
W.~Cary Huffman and Vera Pless.
\newblock {\em Fundamentals of error-correcting codes}.
\newblock Cambridge University Press, Cambridge, 2003.

\bibitem{lopez_2009}
Sergio~R. {L\'opez-Permouth} and Steve Szabo.
\newblock On the {H}amming weight of repeated root cyclic and negacyclic codes
  over {G}alois rings.
\newblock {\em Adv. Math. Commun.}, 3(4):409--420, 2009.

\bibitem{macwilliams_1977_1}
F.~J. MacWilliams and N.~J.~A. Sloane.
\newblock {\em The theory of error-correcting codes. {I}}.
\newblock North-Holland Publishing Co., Amsterdam, 1977.
\newblock North-Holland Mathematical Library, Vol. 16.

\bibitem{Martinez-Moro_2007}
E.~Mart{\'{\i}}nez-Moro and I.~F. R{\'u}a.
\newblock On repeated-root multivariable codes over a finite chain ring.
\newblock {\em Des. Codes Cryptogr.}, 45(2):219--227, 2007.

\bibitem{martinez_perez_2001}
Conchita Mart{\'{\i}}nez-P{\'e}rez, Hans~Georg Schaathun, and Wolfgang Willems.
\newblock On weight hierarchies of product codes. {T}he {W}ei-{Y}ang conjecture
  and more.
\newblock In {\em International {W}orkshop on {C}oding and {C}ryptography
  ({P}aris, 2001)}, volume~6 of {\em Electron. Notes Discrete Math.}, page 7
  pp. (electronic). Elsevier, Amsterdam, 2001.

\bibitem{MASSEY_1973}
James~L. Massey, Daniel~J. Costello, and J{\o}rn Justesen.
\newblock Polynomial weights and code constructions.
\newblock {\em IEEE Trans. Information Theory}, IT-19:101--110, 1973.

\bibitem{ozadam_2009_2}
Hakan {\"O}zadam and Ferruh {\"O}zbudak.
\newblock A note on negacyclic and cyclic codes of length {$p^s$} over a finite
  field of characteristic {$p$}.
\newblock {\em Adv. Math. Commun.}, 3(3):265--271, 2009.

\bibitem{schaathun_2000}
Hans~Georg Schaathun.
\newblock The weight hierarchy of product codes.
\newblock {\em IEEE Trans. Inform. Theory}, 46(7):2648--2651, 2000.

\bibitem{vanlint_1991}
J.~H. van Lint.
\newblock Repeated-root cyclic codes.
\newblock {\em IEEE Trans. Inform. Theory}, 37(2):343--345, 1991.

\bibitem{Wei_Generaized_HW_1991}
Victor~K. Wei.
\newblock Generalized {H}amming weights for linear codes.
\newblock {\em IEEE Trans. Inform. Theory}, 37(5):1412--1418, 1991.

\bibitem{WEI_YANG_Conjecture}
Victor~K. Wei and Kyeongcheol Yang.
\newblock On the generalized {H}amming weights of product codes.
\newblock {\em IEEE Trans. Inform. Theory}, 39(5):1709--1713, 1993.

\end{thebibliography}
\end{document}